\documentclass[envcountsame,envcountsect,orivec,12pt,undated,nolinenumbers]{lnthi}
\usepackage{pstricks,amsmath}

\newenvironment{itemizewithlabel}
 {\begin{itemize}\def\makelabel##1{\hbox to 6pt{\hss\llap{##1}}}}
 {\end{itemize}}

\def\Tlogstar{\log^*\!}
\newbox\Tbox
\def\cev#1{\setbox\Tbox\hbox{$#1$}
  #1\smash{\llap{\raise 1.1pt\hbox{\raise\ht\Tbox\hbox{$\rotatedown{
    \mskip -2.8mu\vec{\smash{\phantom{#1}}}}$}}}}}
\def\Tsim{\equiv}

\title{Space-Efficient DFS and Applications:\\
 Simpler, Leaner, Faster}

\author{Torben Hagerup}

\institute{\Tinfuna[5]\\
  \email{hagerup@informatik.uni-augsburg.de}}

\begin{document}
  \overfullrule=5pt

\maketitle{}

\begin{abstract}
The problem of space-efficient
depth-first search (DFS) is reconsidered.
A particularly simple and
fast algorithm is presented that, on a
directed or undirected input graph $G=(V,E)$ with
$n$ vertices and $m$ edges, carries out a
DFS in $O(n+m)$ time with
$n+\sum_{v\in V_{\ge 3}}\Tceil{\log_2(d_v-1)}
 +O(\log n)\le n+m+O(\log n)$ bits of working memory,
where $d_v$ is the (total) degree of $v$, for
each $v\in V$, and
$V_{\ge 3}=\{v\in V\mid d_v\ge 3\}$.
A slightly more complicated variant of the
algorithm works in the same time with
at most $n+({4/5})m+O(\log n)$ bits.
It is also shown that a DFS can be carried out
in a graph with $n$ vertices and $m$ edges
in $O(n+m\Tlogstar n)$ time with $O(n)$ bits
or in $O(n+m)$ time with
either $O(n\log\log(4+{m/n}))$ bits or,
for arbitrary integer $k\ge 1$,
$O(n\log^{(k)}\! n)$ bits.
These results among them subsume or improve most
earlier results on space-efficient DFS.
Some of the new time
and space bounds are shown to
extend to applications of DFS such as
the computation of cut vertices,
bridges, biconnected components and
2-edge-connected components in undirected graphs.

\bigskip

\end{abstract}

\pagestyle{plain}
\thispagestyle{plain}

\section{Introduction and Related Work}
\label{sec:intro}

Depth-first search or DFS is a very well-known method
for visiting the vertices and edges of a
directed or undirected graph~\cite{CorLRS09,Tar72}.
DFS is set off from other ways of traversing the
graph such as breadth-first search by the
DFS rule:
Whenever two or more vertices were discovered
by the search and have unexplored incident
(out)edges, an (out)edge
incident on the most recently discovered such
vertex is explored first.
The DFS rule confers a number of structural
properties on the resulting graph traversal
that cause DFS to have a large number of
applications.
The rule can be implemented with the aid of a stack
that contains those vertices discovered by the
search that still have unexplored incident (out)edges,
with more recently discovered vertices
being located closer to the top of the stack.
The stack is the main obstacle to a
space-efficient implementation of DFS.

In the following discussion, let $n$ and $m$
denote the number of vertices and of edges,
respectively, of an input graph.
Let us also use the common picture according to
which every vertex is initially \emph{white},
becomes \emph{gray} when it is discovered and
pushed on the stack, and turns \emph{black}
when all its incident (out)edges have been explored
and it leaves the stack.
The study of space-efficient DFS was initiated by 
Asano et al.~\cite{AsaIKKOOSTU14}.
Besides a number of DFS algorithms whose running times
were characterized only as polynomial in $n$ or
worse, they described an algorithm that uses
$O(m\log n)$ time and $O(n)$ bits and another
algorithm that uses $O(n m)$ time and at most
$(\log 3+\epsilon)n$ bits, for arbitrary fixed $\epsilon>0$,
where ``$\log$'', here and in the remainder of
the paper, denotes the binary logarithm
function $\log_2$.
Their basic idea was, since the stack of gray
vertices cannot be kept in full
(it might occupy $\Theta(n\log n)$ bits),
to drop (forget) stack entries and to
restore them in smaller or bigger chunks
when they are later needed.
Using the same idea,
Elmasry, Hagerup and Kammer~\cite{ElmHK15} observed that one
can obtain the best of both algorithms, namely
a running time of $O((n+m)\log n)$ with
$(\log 3+\epsilon)n$ bits.
Assuming a slightly stronger representation of the
input graph as a set of adjacency arrays rather
than adjacency lists, they also devised an
algorithm that runs in $O(n+m)$ time with
$O(n\log\log n)$ bits or in $O((n+m)\log\log n)$
time with $O(n)$ bits, or anything in between
with the same time-space product.
The new idea necessary to obtain this result was,
rather than to forget stack entries entirely,
to keep for each gray vertex a little information
about its entry on the stack and a little
information about
the position of that stack entry.

The space bounds cited so far may be characterized
as \emph{density-independent} in that they depend
only on $n$ and not on~$m$.
If one is willing to settle for
\emph{density-dependent} space bounds that depend
on $m$ or perhaps on the multiset of vertex degrees,
it becomes feasible to store with each gray vertex $u$
an indication of the vertex immediately
above it on the stack,
which is necessarily a neighbor of $u$ and
therefore expressible in $O(\log(d+1))$ bits,
where $d$ is the degree of~$u$.
Since $\log(d+1)=O(d+1)$, this yields a DFS
algorithm that works in $O(n+m)$ time with
$O(n+m)$ bits, as observed
in~\cite{BanCR16,KamKL16}.
One can also use Jensen's inequality to bound the
space requirements of the pointers to
neighboring vertices by $O(n\log(2+{m/n}))$ bits.
This was done in \cite{ElmHK15} for problems for
which the authors were unable to obtain
density-independent bounds.
In the context of DFS, it was mentioned by
Chakraborty, Raman and Satti~\cite{ChaRS16}.

Several applications of DFS relevant to the present
paper can be characterized by means of equivalence
relations on vertices or edges.
Let $G=(V,E)$ be a graph.
If $G$ is directed and $u,v\in V$, let us write
$u\Tsim_G^{\mathrm{S}}v$ if
$G$ contains a path from $u$ to $v$
and one from $v$ to~$u$.
If $G$ is undirected and $e_1,e_2\in E$, write
$e_1\Tsim_G^{\mathrm{B}}e_2$
($e_1\Tsim_G^{\mathrm{E}}e_2$, respectively)
if $e_1=e_2$ or $e_1$ and $e_2$ belong to a common
simple cycle (a not necessarily simple cycle,
respectively) in $G$.
Then $\Tsim_G^{\mathrm{S}}$ is an equivalence
relation on $V$ and $\Tsim_G^{\mathrm{B}}$ and
$\Tsim_G^{\mathrm{E}}$ are equivalence relations on $E$.
Each subgraph induced by an equivalence class of
one of these relations is called a
\emph{strongly connected component}
(\emph{SCC}) in the case
of $\Tsim_G^{\mathrm{S}}$,
a \emph{biconnected component} (\emph{BCC})
or \emph{block}
in the case of $\Tsim_G^{\mathrm{B}}$,
and a \emph{$2$-edge-connected component}
(which we shall abbreviate to \emph{2ECC})
in the case of $\Tsim_G^{\mathrm{E}}$.
Sometimes a single edge with its endpoints is
not considered a biconnected or
2-edge-connected component;
adapting our algorithms to alternative definitions
that differ in this respect is a trivial matter.
Suppose that $G$ is undirected.
A \emph{cut vertex} (also known as an
\emph{articulation point}) in $G$ is a vertex
that belongs to more than one BCC in~$G$;
equivalently, it is a vertex whose removal from $G$
increases the number of connected components.
A \emph{bridge} in $G$ is an edge that belongs
to no cycle in~$G$;
equivalently, it is an edge whose removal from $G$
increases the number of connected components.

For each of the three kinds of components
introduced above, we may want the components
of an input graph to be output one by one.
Correspondingly, we will speak of the SCC,
the BCC and the 2ECC problems.
Outputting a component may mean outputting
its vertices or edges or both.
Correspondingly, we may describe an algorithm
as, e.g., computing the
strongly connected components of a graph
with their vertices.
We may either output special separator symbols
between consecutive components or number the
components consecutively and output
vertices and edges together
with their component numbers;
for our purposes, these two conventions are
equivalent.
\emph{Topologically sorting} a directed acyclic
graph $G=(V,E)$ means outputting the vertices of $G$
in an order such that for each $(u,v)\in E$,
$u$ is output before~$v$.

Elmasry et al.~\cite{ElmHK15} gave algorithms for
the SCC problem and
for topological sorting that work in $O(n+m)$ time using
$O(n\log\log n)$ bits.
Their main tool was a method for
``coarse-grained reversal'' of a DFS computation
that makes it possible to output the vertices of
the input graph in \emph{reverse postorder}, i.e.,
in the reverse of the order in which the vertices
turn black in the course of the DFS.
Various bounds for these problems were claimed
without proof by
Banerjee, Chakraborty and Raman~\cite{BanCR16}:
$O(m\log n\log\log n)$ time with $O(n)$ bits
for the SCC problem and
$O(n+m)$ time with $m+3 n+o(n+m)$ bits as well as
$O(m\log\log n)$ time with $O(n)$ bits
for topological sorting.
For the BCC problem and the computation of
cut vertices, Kammer,
Kratsch and Laudahn~\cite{KamKL16}
described an algorithm that works in $O(n+m)$
time using $O(n+m)$ bits and can be seen as
an implementation of an algorithm of
Schmidt~\cite{Sch13}.
Essentially the same algorithm was sketched by
Banerjee et al.~\cite{BanCR16}, who also
applied it to the 2ECC problem and the
computation of bridges.
Space bounds of the form $O(n\log({m/n}))$ for
the same problems were mentioned by
Chakraborty et al.~\cite{ChaRS16}.
Essentially re-inventing an algorithm
of Gabow~\cite{Gab00} and combining it with
machinery from~\cite{ElmHK15} and with new ideas,
Kammer et al.~\cite{KamKL16}
also demonstrated how to compute the cut vertices
in $O(n+m)$ time with $O(n\log\log n)$ bits.
Finally, decomposing the input graph into
subtrees and processing the subtrees one by one,
Chakraborty et al.~\cite{ChaRS16}
were able to solve the BCC problem and
compute the cut vertices in
$O(m\log n\log\log n)$ time with $O(n)$ bits.

\section{New Results and Techniques}
\label{sec:techniques}

The main thrust of this work is to establish
new density-dependent and density-independent
space bounds for fast DFS algorithms.
Let us begin by developing simple notation
that allows the results to be stated conveniently.

When $G=(V,E)$ is a directed or undirected graph,
$d_v$ is the (total) degree of $v$ for each $v\in V$
and $k$ is an integer, let
\[
L_k(G)=\sum_{{v\in V}\atop{d_v+k\ge 2}}\Tceil{\log_2(d_v+k)}.
\]
When $G$ is directed, we use $L_k^{\mathrm{in}}(G)$ and
$L_k^{\mathrm{out}}(G)$ to denote quantities defined in
the same way, but now with $d_v$ taken to mean the indegree
and the outdegree of~$v$, respectively.

\begin{lemma}
\label{lem:tolog}
Let $G$ be a directed or undirected graph with
$n$ vertices and $m$ edges.
Then
\begin{itemizewithlabel}
\item[(a)]
$L_1(G)\le n\log(1+{{4 m}/n})$;
\item[(b)]
If $G$ is directed, then
$L_1^{\mathrm{in}}(G)$ and $L_1^{\mathrm{out}}(G)$
are both bounded by $n\log(1+{{2 m}/n})$.
\end{itemizewithlabel}
\end{lemma}

\begin{proof}
Let $G=(V,E)$ and, for each $v\in V$,
denote by $d_v$ the (total) degree of~$v$.
To prove part~(a), observe first that
$\Tceil{\log(d+1)}\le\log(2 d+1)$ for all
integers $d\ge 0$.
Since the function $d\mapsto\log(2 d+1)$ is
concave on $[0,\infty)$ and
$\sum_{v\in V}d_v=2 m$,
the result follows from
Jensen's inequality.
Part~(b) is proved in the same way, noting that
the relevant vertex degrees now sum to~$m$.
\end{proof}

Our most accurate space bounds involve terms of
the form $L_k(G)$.
More convenient bounds can be derived from them
with Lemma~\ref{lem:tolog}.
Note that for all $a,b,c>0$, the quantity
$a n+n\log(b+{{c m}/n})$ can also
be written as $n\log(2^a b+{{2^a c m}/n})$.
The latter form will be preferred here.

Our first algorithm carries out a DFS of a
graph $G$ with $n$ vertices and $m$ edges
in $O(n+m)$ time using at most
$n+L_{-1}(G)+O(\log n)$ bits.
The number of bits needed, which can also be bounded
by $n+m+O(\log n)$ and by $n\log(2+{{8 m}/n})+O(\log n)$,
is noteworthy only for the constant
factors involved.
Comparable earlier space bounds were indicated
only as $O(n+m)$ or $O(n\log({m/n}))$ bits, and
no argument offered in their support
points to as small constant factors as ours.
Moreover, all of the earlier algorithms make use of
rank-select structures~\cite{Cla96},
namely to store variable-length
information indexed by vertex numbers.
Whereas asymptotically space-efficient
and fast rank-select structures
are known, it is generally
accepted that in practice
they come at a considerable price in
terms of time and especially space
(see, e.g.,~\cite{Vig08}) and a certain
coding complexity.
In contrast, we view the algorithm
presented here
as the first truly practical
space-efficent DFS algorithm.

The simple but novel idea that enables us to make
do without rank-select structures is a different
organization of the DFS stack.
The vertices on the stack, in the order from
the bottom to the top of the stack, always
form a directed path, in $G$ itself if $G$ is
directed and in the directed version of $G$ if not,
that we call the \emph{gray path}.
Assume that $G$ is undirected.
Instead of having a table that maps each vertex
to how far it has progressed in the exploration
of its incident edges, which in some sense
distributes the stack over the single vertices
and is what necessitates a rank-select structure,
we return to using a stack implemented in
contiguous memory locations and store there for each
internal vertex $v$ on the gray path
the distance in its
adjacency array, considered as a cyclic
structure, from the predecessor $u$ of~$v$
to the successor $w$ of $v$ on the gray path.
More intuitively, one can think of the stack
entry as describing the ``turn'' that the
gray path makes at $v$, namely
from $u$ via $v$ to $w$.
Knowing $w$, $v$ and the ``turn value'', one can
compute $u$.
Provided that outside of the stack we always remember the
\emph{current vertex} $w$ of the DFS,
the vertex on top of the DFS stack
and at the end of the gray path,
and the position in $w$'s adjacency array
of the predecessor of $w$ on the gray path, if any,
this allows us to pop from the stack
in constant time, and pushing is equally easy.
In the course of the processing of~$v$,
the ``turn value'' can be stepped from~1
(``after entering $v$ from $u$, take the next exit'')
to $d-1$, where $d$ is the (total) degree of~$v$
(directed edges that enter~$v$ are simply ignored).
Aside from the somewhat unusual stack, the DFS
can proceed as a usual DFS and
complete in linear time.
Handling vertices of small degree specially, we
can lower the space bound to
$n+({4/5})m+O(\log n)$ bits and solve the
SCC problem in $O(n+m)$ time with
$n\log 3+({{14}/5})m+O((\log n)^2)$ bits.
Resorting to using rank-select structures, we
describe linear-time algorithms for the SCC, BCC
and 2ECC problems and for the computation of
topological sortings, cut vertices and bridges
with space bounds of the form
$(a n+b m)(1+o(1))$ or
$a n\log(b+{{c m}/n})(1+o(1))$ bits,
where $a$, $b$ and $c$ are positive constants.
Apart from minor tricks to reduce
the values of $a$, $b$ and $c$, 
no new techniques are involved here.

Turning to space bounds that are independent
on $m$ or almost so, we first describe a DFS
algorithm that works in $O(n+m)$ time with
$O(n\log\log(4+{m/n}))$ bits.
The algorithm is similar to an algorithm of
Elmasry et al.~\cite{ElmHK15}
that uses $\Theta(n\log\log n)$ bits.
Our superior space bound is made possible by
two new elements:
First, the algorithm is changed to use a stack
of ``turn values'' rather than of
``progress counters'', as discussed above.
And second, when stack entries have to be dropped
to save space, we keep approximations of the
lost entries that turn out to work better than
those employed in~\cite{ElmHK15}.
Our space bound is attractive because it unifies the
earlier bounds of the forms $O(n+m)$,
$O(n\log(2+{m/n}))$ and $O(n\log\log n)$ bits,
being at least as good as all of them
for every graph density and better than each
of them for some densities.

Subsequently we show how to carry out
a DFS in $O(n+m\Tlogstar n)$ time with $O(n)$ bits or,
with a slight variation, in $O(n+m)$ time
with $O(n\log^{(k)}\! n)$ bits for arbitrary
fixed $k\in\TbbbN=\{1,2,\ldots\}$.
Here $\log^{(k)}\!$ denotes $k$-fold repeated
application of $\log$, e.g., $\log^{(2)}\! n=\log\log n$,
and $\Tlogstar n=\min\{k\in\TbbbN\mid\log^{(k)}\! n\le 1\}$.
The main new idea instrumental in obtaining
this result is to let each vertex $v$ dropped from
the stack record, instead of a fixed approximation
of its stack position as in earlier algorithms,
an approximation of that position that changes
dynamically to become coarser when $v$ is farther
removed from the top of the stack.
Adapting an algorithm of Kammer et al.~\cite{KamKL16}
for computing cut vertices, we show that
the time and space bounds indicated in this paragraph
extend to the problems of computing biconnected
and 2-edge-connected components, cut vertices
and bridges of undirected graphs.

\section{Preliminaries}
\label{sec:prelim}

We assume a representation of an undirected
input graph $G=(V,E)$ that
is practically identical to the one
used in~\cite{HagKL17}:
For some known integer $n\ge 1$, $V=\{1,\ldots,n\}$,
the degree of each $u\in V$ can be obtained
as $\Tvn{deg}(u)$, and for each
$u\in V$ and each $i\in\{0,\ldots,\Tvn{deg}(u)-1\}$,
$\Tvn{head}(u,i)$ and $\Tvn{mate}(u,i)$ yield
the $(i+1)$st neighbor $v$ of $u$ and the
integer $j$ with $\Tvn{head}(v,j)=u$,
respectively, for some arbitrary
numbering, starting at 0,
ordering of the neighbors of each vertex.
The access functions $\Tvn{deg}$, $\Tvn{head}$
and $\Tvn{mate}$ run in constant time.
The representation of a directed graph
$G$ is similar in spirit:
$G$ is represented with in/out adjacency arrays,
i.e., we can access the inneighbors as well
as the outneighbors of a given vertex one by one,
and there are \emph{cross links}, i.e., the
function $\Tvn{mate}$ now, for each edge $(u,v)$,
maps the position of $v$ in the adjacency array
of $u$ to that of $u$ in the adjacency array
of $v$ and vice versa.

A DFS of a graph $G$ is associated with a
spanning forest $F$ of $G$ in an obvious way:
If a vertex $v$ is discovered by the DFS when
the current vertex is $u$,
$v$ becomes a child of~$u$.
$F$ is called the \emph{DFS forest} corresponding
to the DFS, and its edges are called
\emph{tree edges}, whereas the other edges of $G$
may be called \emph{nontree edges}.
At the outermost level, the DFS steps through
the vertices of $G$ in a particular order, called
its \emph{root order}, and every vertex found
not to have been discovered at that time
becomes the root of the next \emph{DFS tree} in~$F$.
The \emph{parent pointer} of a given vertex $v$
in $G$ is an indication of the parent $u$ of
$v$ in $F$, if any.
If the root order of a DFS is simply $1,\ldots,n$
and the DFS always explores the edges incident on
the current vertex $u$ in the order in which
their endpoints occur in the adjacency array
of $u$, the corresponding DFS forest
is the \emph{lexicographic} DFS forest of the
adjacency-array representation.

The following lemmas describe two auxiliary
data structures that we use repeatedly:\
the \emph{choice dictionary} of
Kammer and
Hagerup~\cite{Hag17,HagK16}
and the \emph{ternary array} of
Dodis, P\v atra\c scu and Thorup~\cite[Theorem~1]{DodPT10}.

\begin{lemma}
\label{lem:choice}
There is a data structure that,
for every $n\in\TbbbN$,
can be initialized
for \emph{universe size} $n$
in constant time
and subsequently occupies $n+O({n/{\log n}})$ bits
and maintains an initially empty subset $S$ of
$\{1,\ldots,n\}$ under insertion,
deletion, membership queries and the operation
\Tvn{choice} (return an arbitrary element of $S$)
in constant time
as well as iteration over $S$ in $O(|S|+1)$ time.
\end{lemma}

\begin{lemma}
\label{lem:ternary}
There is a data structure that
can be initialized with an arbitrary $n\in\TbbbN$
in $O(\log n)$ time
and subsequently occupies $n\log_2 3+O((\log n)^2)$ bits
and maintains a sequence drawn from
$\{0,1,2\}^n$ under constant-time reading and
writing of individual elements of the sequence.
\end{lemma}

\section{Density-Dependent Bounds}

\subsection{Depth-First Search}

\begin{theorem}
\label{thm:dfs}
A DFS of a directed or undirected graph $G=(V,E)$
with $n$ vertices and $m$ edges can be
carried out in $O(n+m)$ time with 
at most any of the following
numbers of bits of working memory:
\begin{itemizewithlabel}
\item[(a)]
$n+L_{-1}(G)+O(\log n)$;
\item[(b)]
$n+m+O(\log n)$;
\item[(c)]
$n\log(2+{{8 m}/n})+O(\log n)$.
\end{itemizewithlabel}
\end{theorem}

\begin{proof}
We first show part~(a) for the case in
which $G$ is undirected.
The algorithm was described in
Section~\ref{sec:techniques}, and it was argued there
that it works in $O(n+m)$ time.
What remains is to bound the number of
bits needed.

If an internal vertex $v$ on the gray path
has degree~$d$, its stack entry can be taken
to be an integer in $\{1,\ldots,d-1\}$ that
indicates the number of edges incident
on $v$ that were explored with $v$ as
the current vertex.
The stack entry can therefore be represented
in $\Tceil{\log(d-1)}$ bits, so that the
entire stack never occupies more than
$L_{-1}(G)$ bits.
In addition to the information on the stack,
the DFS must know for each vertex $v$
whether $v$ is white;
this takes $n$ bits.
(Unless an application calls for it, the
DFS has no need to distinguish
between gray and black vertices.)
Finally the DFS must store a few simple
variables in $O(\log n)$ bits, for a grand total
of $n+L_{-1}(G)+O(\log n)$ bits.
This concludes the proof of part~(a)
for undirected graphs.

If $G$ is directed, we can pretend that the
inneighbors and the outneighbors of each vertex
are stored in the same adjacency array
(whether or not this is the case in the actual
representation of~$G$).
We can then use the same algorithm, except that
an edge $(u,v)$ should not be explored in the
wrong direction, i.e., when $v$ is the
current vertex of the DFS.

To show part~(b) of the theorem,
let $d_v$ be the (total) degree of $v$
for each $v\in V$ and observe that
$\Tceil{\log(d-1)}\le {d/2}$
for all integers $d\ge 3$, so that
$L_{-1}(G)\le ({1/2})\sum_{v\in V} d_v=m$.
Part~(c) follows immediately
from part~(a) by an application
of Lemma~\ref{lem:tolog}(a).
\end{proof}

At the price of introducing a slight complication
in the algorithm, we can obtain another space bound
of $c_1 n+c_2 m+O(\log n)$ bits for a smaller
constant $c_2$.
If $c_1$ is allowed to increase, it is also possible
(but of little interest) to lower $c_2$ as far as
desired towards~0 by treating vertices
of small degree separately
in the analysis.

\begin{theorem}
\label{thm:four}
A DFS of a directed or undirected graph $G$
with $n$ vertices and $m$ edges can be
carried out in $O(n+m)$ time with 
at most $n+({4/5})m+O(\log n)$ bits of working memory.
\end{theorem}

\begin{proof}
The relation $\Tceil{\log(d-1)}\le ({2/5})d$
is satisfied for all integers $d\ge 3$
except 4, 6 and 7.
To handle the stack entries of vertices of
degree~4, we divide these into groups of 5
and represent each group on the stack
through a single combined
entry of $\Tceil{\log((4-1)^5)}=8$ bits
instead of 5 individual
entries of $\Tceil{\log(4-1)}=2$ bits each.
Since $8\le({2/5})\cdot 5\cdot 4$, the combined
entry is small enough for the bound of the theorem.
At all times, an incomplete group of up to
4 individual entries is kept outside of the
stack in a constant number of bits.
Similarly, groups of 3 entries for vertices
of degree~6 are represented in
$\Tceil{\log((6-1)^3)}=7$ bits,
and groups of 3 entries for vertices
of degree~7 are represented in
$\Tceil{\log((7-1)^3)}=8$ bits.
Since $7\le({2/5)}\cdot 3\cdot 6$
and $8\le({2/5)}\cdot 3\cdot 7$,
this altogether yields a space bound of
$n+({2/5})\cdot 2 m+O(\log n)$ bits.
\end{proof}

The simplicity of the
algorithm of Theorem~\ref{thm:dfs}
is demonstrated in Fig.~\ref{fig:dfs}, which
shows 
an implementation of it for
an undirected input graph $G=(V,E)$.
The description is given in complete
detail except for items like
the declaration of variables
and for the
specification of a bit stack $S$ with the
following two operations in addition to an
appropriate initialization to being empty:
$S.\Tvn{push}(\ell,d)$, where $\ell$ and $d$ are
integers with $d\ge 3$ and $0\le \ell<2^{\Tceil{\log(d-1)}}$,
pushes on $S$ the
$\Tceil{\log(d-1)}$-bit binary representation
of $\ell$, and $S.\Tvn{pop}(d)$,
where $d$ again is an integer with $d\ge 3$,
correspondingly
pops $\Tceil{\log(d-1)}$ bits from $S$,
interprets these as the binary representation
of an integer $\ell$ and returns~$\ell$.
The task of the DFS is assumed to be the execution
of certain \emph{user procedures} at
the appropriate times:
$\Tvn{preprocess}(v)$ and $\Tvn{postprocess}(v)$,
for each $v\in V$, when $v$ turns gray and
when it turns black, respectively,
$\Tvn{explore\_tree\_edge}(v,w)$, for
$\{v,w\}\in E$, when the edge $\{v,w\}$ is
explored with $v$ as the current vertex
and becomes a tree edge,
$\Tvn{retreat\_tree\_edge}(v,w)$ when the
DFS later withdraws from $w$ to~$v$,
and $\Tvn{handle\_back\_edge}(v,w)$,
for $\{v,w\}\in E$, when $\{v,w\}$ is explored
with $v$ as the current vertex
but does not lead to a new vertex.
The code is made slightly more involved
by a special handling of the first and
last vertices of the gray path and by the
fact that no stack entries are
stored for vertices of degree~2.
Timing experiments with an implementation
of the algorithm of Fig.~\ref{fig:dfs} showed it to
be sometimes faster and sometimes slower than an
alternative algorithm that also manages its own
stack but makes no attempt at being space-efficient.

\begin{figure}
\begin{tabbing}
\quad\=\quad\=\quad\=\quad\=\quad\=\hskip 8cm\=\kill
DFS:\\
\>\textbf{for} $v\in V$ \textbf{do} $\Tvn{white}[v]:=\Tvn{true}$;
 $(*$ initially all vertices are undiscovered $*)$\\
\>\textbf{for} $v\in V$ \textbf{do if} $\Tvn{white}[v]$ \textbf{then}
 $(*$ if $v$ has not yet been discovered $*)$\\
\>\>$\Tvn{root}:=v$;
 $(*$ begin a new DFS tree rooted at $v$ $*)$\\
\>\>$k:=-1$; $\ell:=-1$;\\
\>\>$\Tvn{white}[v]:=\Tvn{false}$;\\
\>\>$\Tvn{preprocess}(v)$;
 $(*$ $v=\Tvn{root}$ $*)$\\
\>\>\textbf{repeat}
 $(*$ until breaking out of the loop with \textbf{break} below $*)$\\
\>\>\>$(*$ Invariant: $v$ is the current vertex,
 with data on $v$ stored in $k$ and $\ell$ $*)$\\
\>\>\>$\ell:=\ell+1$;
 $(*$ advance in $v$'s adjacency array $*)$\\
\>\>\>\textbf{if} $\ell<\Tvn{deg}(v)$ \textbf{then}
 $(*$ if $v$ still has unexplored incident edges $*)$\\
\>\>\>\>$w:=\Tvn{head}(v,(k+\ell+1)\bmod\Tvn{deg}(v))$;
 $(*$ the next neighbor of $v$ $*)$\\
\>\>\>\>\textbf{if} $\Tvn{white}[w]$ \textbf{then}\\
\>\>\>\>\>$\Tvn{explore\_tree\_edge}(v,w)$;\\
\>\>\>\>\>\textbf{if} $v=\Tvn{root}$
 \textbf{then} $\Tvn{$\ell$0}:=\ell$;
 $(*$ save $\ell$ in $\Tvn{$\ell$0}$ rather than on $S$ $*)$\\
\>\>\>\>\>\textbf{else if} $\Tvn{deg}(v)>2$
 \textbf{then} $S.\Tvn{push}(\ell,\Tvn{deg}(v))$;
 $(*$ push $\ell$ in nontrivial cases $*)$\\
\>\>\>\>\>$k:=\Tvn{mate}(v,(k+\ell+1)\bmod\Tvn{deg}(v))$;
 $(*$ index at $w$ of $v$ $*)$\\
\>\>\>\>\>$\ell:=-1$; $(*$ prepare to take the first turn out of $w$ $*)$\\
\>\>\>\>\>$v:=w$;
 $(*$ make $w$ the current vertex $*)$\\
\>\>\>\>\>$\Tvn{white}[v]:=\Tvn{false}$;\\
\>\>\>\>\>$\Tvn{preprocess}(v)$;\\
\>\>\>\>\texttt{else} $\Tvn{handle\_back\_edge}(v,w)$;\\
\>\>\>\textbf{else}
 $(*$ $v$ has no more unexplored incident edges $*)$\\
\>\>\>\>\textbf{if} $v=\Tvn{root}$ \textbf{then break};
 $(*$ done at the root -- a DFS tree is finished $*)$\\
\>\>\>\>$u:=\Tvn{head}(v,k)$;
 $(*$ the parent of $v$ in the DFS tree $*)$\\
\>\>\>\>\textbf{if} $u=\Tvn{root}$ \textbf{then} $\ell:=\Tvn{$\ell$0}$;
 $(*$ retrieve $\ell$ from $\Tvn{$\ell$0}$ rather than from $S$ $*)$\\
\>\>\>\>\textbf{else if} $\Tvn{deg}(u)\le 2$ \textbf{then} $\ell:=0$;
 $(*$ trivial case -- nothing stored on $S$ $*)$\\
\>\>\>\>\textbf{else} $\ell:=S.\Tvn{pop}(\Tvn{deg}(u))$;
 $(*$ pop $\ell$ in nontrivial cases $*)$\\
\>\>\>\>$k:=(\Tvn{mate}(v,k)-(\ell+1))\bmod\Tvn{deg}(u)$;
 $(*$ index at $u$ of $u$'s parent $*)$\\
\>\>\>\>$\Tvn{postprocess}(v)$;\\
\>\>\>\>$\Tvn{retreat\_tree\_edge}(u,v)$;\\
\>\>\>\>$v:=u$;
 $(*$ make $u$ the current vertex $*)$\\
\>\>\textbf{forever};\\
\>\>$\Tvn{postprocess}(v)$;
 $(*$ $v=\Tvn{root}$ $*)$
\end{tabbing}
\caption{The algorithm of Theorem~\ref{thm:dfs}
for an undirected input graph $G=(V,E)$.}
\label{fig:dfs}
\end{figure}

\subsection{Strongly Connected Components and
Topological Sorting}

\begin{theorem}
\label{thm:scc1}
The
strongly connected components of a directed
graph with $n$ vertices and $m$ edges can be computed
with their vertices and/or edges
in $O(n+m)$ time with $n\log_2 3+({{14}/5})m+O((\log n)^2)$
bits of working memory.
\end{theorem}

\begin{proof}
Let $G$ be the input graph and let $\cev{G}$ be the
directed graph obtained from $G$ by replacing each
edge $(u,v)$ by the antiparallel edge~$(v,u)$.
We use an algorithm attributed to Kosaraju and
Sharir in~\cite{AhoHU83} that identifies the
vertex set of each
SCC as that of a DFS tree
constructed by a standard DFS of $\cev{G}$
that, however,
employs as its root order the
reverse postorder defined by an (arbitrary) DFS of~$G$.

Consider each vertex $v$ in $G$ to have a
circular incidence array that contains 
all edges entering $v$ as well as all edges leaving $v$.
A DFS of $G$ can be viewed as entering
each nonroot vertex $v$ at a particular (tree) edge
and each root $v$ at a fixed position in its incidence
array and eventually traversing $v$'s incidence
array exactly once from that entry point, classifying
certain edges out of $v$ as tree edges and skipping over
the remaining edges, either because they lead
to vertices that were already discovered or
because they enter~$v$, before finally,
if $v$ is a nonroot,
retreating over the tree edge to $v$'s parent.
During such a DFS of $G$ that uses the
root order $1,\ldots,n$, we construct a bit
sequence $B$ by appending a \texttt{1} to an
initially empty sequence whenever the DFS
discovers a new vertex or withdraws over a
tree edge and by appending a \texttt{0}
whenever the DFS skips over an edge.
The total number of bits in $B$ is exactly $2 m$,
and $B$ can be seen to represent an Euler tour
of each tree in the forest $F$
defined by the DFS in a natural way.
We also use an array $A$
of $n$ bits to mark those vertices
that are roots in~$F$.
Observe that the pair $(A,B)$ supports an
\emph{Euler traversal} that, in $O(n+m)$ time
and using only $O(\log n)$ additional bits,
enumerates the vertices in $G$
in reverse postorder with respect to~$F$.
In particular, whenever an Euler tour of
a tree in $F$ with root $r$ has been followed backwards
completely from end to start,
$A$ is used to find the end vertex of the next Euler tour,
if any, as the largest root smaller than~$r$.

We carry out a DFS of $\cev{G}$, interleaved with an
execution of the Euler traversal that supplies
new root vertices as needed,
and output the vertex set
of each resulting tree as an SCC.
The total time spent is $O(n+m)$.
We could execute the algorithm using
$n$ bits for $A$,
$2 m$ bits for $B$ and, according to
Theorem~\ref{thm:four}, $n+({4/5})m+O(\log n)$
bits for the depth-first searches.
Recall, however, that the space bound of
Theorem~\ref{thm:four} is obtained as the sum of $n$ bits
for an array $\Tvn{white}$ and
$({4/5})m+O(\log n)$ bits for the DFS stack
and related variables.
It turns out that we can realize $A$ and \Tvn{white}
together through a single ternary array
with $n$ entries.
To see this, it suffices in the case of the DFS of $G$ to note
that a vertex classified as a root certainly is not white.
For the DFS of $\cev{G}$, assume first that we want
to output only the vertices of the strongly
connected components, as is standard.
Then even a binary array would
suffice---we could use the same binary value to
denote both ``root'' and ``not white''.
The reason for this is,
on the one hand, that
when the Euler traversal has entered a tree
$T$ with root~$r$, it will never again
need to inspect $A[w]$ for any $w\ge r$ and,
on the other hand,
that every vertex $w$ reachable in $\cev{G}$
from a vertex $v$ in $T$ must satisfy $w\ge r$---otherwise
$v$ would belong to an earlier tree
(with respect to the DFS of~$G$) and not to~$T$.

If we want to output not only the vertices,
but also the edges of each SCC and
perhaps to highlight those edges whose endpoints
belong to different strongly connected
components (the ``inter-component'' edges),
we need to know for each edge $(v,w)$
explored during the DFS of $\cev{G}$
whether $w$ belongs to the
DFS tree under construction at that time
(then $(v,w)$ is an edge of the current SCC)
or to an older DFS tree
(then $(v,w)$ is an ``inter-component'' edge).
We solve this problem again resorting to a
ternary array, splitting the value ``not white'' into
``not white, but in the current tree'' and
``in an older tree''.
Whenever the DFS of $\cev{G}$ completes a tree,
we repeat the DFS of that tree, treating the color
``not white, but in the current tree'' as ``white''
and replacing all its occurrences by
``in an older tree''.
The space bound follows from Lemma~\ref{lem:ternary}.
\end{proof}

When $m$ is larger relative to~$n$,
it is advantageous, instead
of storing the bit vector $B$, to store for each
vertex $v$ a \emph{parent pointer} of
$\Tceil{\log(d+1)}$ bits, where $d$ is the
indegree of $v$, that indicates $v$'s parent
in the DFS forest of $G$ or no parent at all
(i.e., $v$ is a root).
For this we need the standard
\emph{static space allocation}:

\begin{lemma}
\label{lem:static}
There is a data structure that can be initialized
for a positive integer $n$ and $n$ nonnegative integers
$\ell_1,\ldots,\ell_n$ in $O(n+N)$ time, where
$\ell_j=O(\log n)$ for $j=1,\ldots,n$ and $N=\sum_{j=1}^n \ell_j$,
and subsequently occupies $(n+2 N)(1+O({{\log\log n}/{\log n}}))$
bits and realizes an array $A[1\Ttwodots n]$ of
entries of $\ell_1,\ldots,\ell_n$ bits
under constant-time reading
and writing of individual entries in~$A$.
\end{lemma}

\begin{proof}
Maintain the entries of $A$ in an array $\widetilde{A}[0\Ttwodots N-1]$
of $N$ bits
and store the sequence $B=b_1\cdots b_{n+N}=
\texttt{0}^{\ell_1}\texttt{1}\cdots\texttt{0}^{\ell_n}\texttt{1}$
of $n+N$ bits.
For $k=1,\ldots,n$, $A[k]$ is located in
$\widetilde{A}[\Tvn{select}_B(k-1)-(k-1)\Ttwodots\Tvn{select}_B(k)-k-1]$,
where
$\Tvn{select}_B(k)=\min\{j\in\{0,\ldots,n+N\}\mid\sum_{i=1}^j b_i=k\}$
for $k=0,\ldots,n$,
and $\Tvn{select}_B$ can be evaluated in constant time given
$O({{(n+N)\log\log n}/{\log n}})$ bits of
bookkeeping information~\cite{Gol07,RamRS07}
that can be computed in $O(n+N)$ time.
\end{proof}

\begin{lemma}
\label{lem:parent}
A representation of
the parent pointers of the lexicographic DFS forest of an
adjacency-array representation of a
graph $G$ with $n$ vertices and $m$ edges that allows
constant-time access to the parent of a given vertex
can be stored in
$(n+2 N)(1+O({{\log\log n}/{\log n}}))$ bits
and computed in $O(n+m)$ time with $n$ additional bits,
where $N=L_1(G)$ if $G$ is undirected and
$N=L^{\mathrm in}_1(G)$ if $G$ is directed.
\end{lemma}

\begin{proof}
The parent pointers themselves can be stored
in $N$ bits, and constant-time access to them can be
provided according to Lemma~\ref{lem:static}.
To compute the parent pointers, carry out a DFS of $G$,
using the $n$ additional bits to store for
each vertex whether it is still white.
When the DFS ends the processing at a vertex $v$,
it follows the parent pointer of $v$ to withdraw to
$v$'s parent $u$ in the DFS forest,
and from there proceeds to explore
the edge that follows $(u,v)$ or $\{u,v\}$ in
$u$'s incidence array, if any, and to store the
appropriate new parent
pointer if this edge leads to a white vertex.
The procedure to follow at the first
exploration of an edge from a newly discovered
vertex is analogous.
\end{proof}

\begin{theorem}
\label{thm:scc2}
The strongly connected components of a directed
graph $G=(V,E)$ with $n$ vertices and $m$ edges can be computed
in $O(n+m)$ time with at most
$(2 n+L_{-1}(G)+2 L^{\mathrm{in}}_1(G))(1+O({{\log\log n}/{\log n}}))
\le 3 n\log(2+{{4 m}/n})(1+O({{\log\log n}/{\log n}}))$
bits of working memory.
\end{theorem}

\begin{proof}
The parent pointers of a DFS of $G$ by themselves
support the Euler traversal of the proof
of Theorem~\ref{thm:scc1} in $O(n+m)$ time,
using $O(\log n)$ additional bits.
To see this, observe that one can
visit the children of a vertex $u$
by inspecting the outneighbors of $u$ one by one
to see which of them indicate $u$
as their parent
and that the array $A$ is superfluous since
a vertex is a root in the DFS forest if and
only if its parent pointer does not point
to one of its neighbors---a value was
reserved for this purpose.
Thus first compute the parent pointers
(Lemma~\ref{lem:parent}) and then carry
out a DFS of $\cev{G}$ interleaved with
the Euler traversal.
The time needed is $O(n+m)$, and the number of bits is
at most the sum of the bounds of
Theorem~\ref{thm:dfs} and Lemma~\ref{lem:parent}.
To prove the second bound,
use Lemma~\ref{lem:tolog}.
\end{proof}

If the input graph $G$ happens to be acyclic,
the algorithms of Theorems
\ref{thm:scc1} and~\ref{thm:scc2}
output the vertices of $G$ in the order of a
topological sorting.
In the case of Theorem~\ref{thm:scc1} this may
yield the most practical algorithm.
Better space bounds for topological sorting
can, however, be obtained by implementing an
alternative standard algorithm, due to
Knuth~\cite{Knu97}, that repeatedly removes
a vertex of indegree~0 while keeping track
only of the indegrees of all vertices.
This was also suggested by Banerjee
et al.~\cite{BanCR16}.
As mentioned in the discussion
of related work,
they indicated a space bound of $m+3 n+o(n+m)$ bits;
it is not clear to this author, however,
how such a bound is to be proved.

\begin{theorem}
A topological sorting of a directed acyclic
input graph with $n$ vertices and $m$ edges can
be computed in $O(n+m)$ time with at most
any of the following numbers of bits:
\begin{itemizewithlabel}
\item[(a)]
$(2 n+2 L_0^{\mathrm{in}})(1+O({{\log\log n}/{\log n}}))$;
\item[(b)]
$(2 n+({4/3})m)(1+O({{\log\log n}/{\log n}}))$;
\item[(c)]
$2 n\log(2+{{4 m}/n})(1+O({{\log\log n}/{\log n}}))$.
\end{itemizewithlabel}
\end{theorem}

\begin{proof}
Maintain the current set of vertices of indegree~0
in an instance of the choice dictionary
of Lemma~\ref{lem:choice},
which needs $n+O({n/{\log n}})$ bits.
Also maintain the
current indegrees according to Lemma~\ref{lem:static}.
Since we can store an arbitrary value or nothing for
vertices of current indegree~0, we need only distinguish
between $d$ different values for a vertex of
original indegree $d\ge 2$, so that
$(n+2 L_0^{\mathrm{in}})(1+O({{\log\log n}/{\log n}}))$
bits suffice.
With these data structures,
the algorithm of
Knuth~\cite{Knu97} can be executed in $O(n+m)$ time.
This proves part~(a).
Part~(b) follows from part~(a) since
$\Tceil{\log d}\le ({2/3})d$ for all integers $d\ge 2$,
and part~(c) follows from part~(a) with
Lemma~\ref{lem:tolog}(b).
\end{proof}

\subsection{Biconnected and 2-Edge-Connected Components}
\label{subsec:dep-apl}
In this subsection we will see that closely related
algorithms can be used to compute the
cut vertices, the bridges and the
biconnected and 2-edge-connected components of
an undirected graph.
Our algorithms are similar to
those of \cite{BanCR16,ChaRS16,KamKL16}, but whereas the
earlier authors indicated the space bounds only
as $O(n+m)$ or $O(n\log({m/n}))$ bits, we will strive to obtain
small constant factors and indicate these explicitly.

A simple but crucial fact is that for every DFS forest
$F$ of an undirected graph $G$, every edge in $G$
joins an ancestor to a descendant within a tree in~$F$.
DFS is also known to interact harmoniously with
the graph structures of interest
in this subsection
as exemplified, e.g., in the following lemma.

\begin{lemma}
\label{lem:subtree}
Let $F$ be a DFS forest of an undirected graph
$G=(V,E)$ and let $e\in E$.
Then the subgraph $F'$ of $F$ induced by the edges
in $F$ equivalent to $e$ under $\Tsim_G^{\textrm{B}}$
is a subtree of $F$ whose root has degree $1$ in~$F'$.
\end{lemma}

\begin{proof}
Every edge in $G$ is equivalent under
$\Tsim_G^{\textrm{B}}$ to an edge in~$F$,
so $F'$ is not the empty graph.
Let us first prove that $F'$ is connected.
Suppose for this that a simple path $\pi$ in $F$
contains the edges $e_1$, $e_2$ and $e_3$
in that order and that $e_1$ and $e_3$
belong to~$F'$.
To show that $e_2$ also belongs to~$F'$, let
$C$ be a simple cycle in $G$ that contains $e_1$ and~$e_3$
and let $\pi'$ be the maximal subpath
of $\pi$ that contains $e_2$ and whose internal
vertices do not belong to~$C$.
The endpoints of $\pi'$ lie on $C$, so $\pi$
and a suitably chosen
subpath of $C$ together form a simple cycle
that contains $e_2$ and at least one of
$e_1$ and $e_3$.
Thus $F'$ is indeed a subtree of~$F$ with a root~$u$.
A simple cycle in $G$ that contains two edges in $F'$ incident
on $u$ must necessarily also contain a proper
ancestor of~$u$, contradicting the
fact that the edge between $u$ and its parent
in $F$, if any,
does not belong to~$F'$.
Thus the degree of $u$ in $F'$ is~1.
\end{proof}

Let $F$ be a DFS forest of an undirected
graph $G=(V,E)$ with $n$ vertices
and $m$ edges.
Let us call a subtree $F'$ of $F$ as in Lemma~\ref{lem:subtree}
a \emph{BCC subtree} and its root a \emph{BCC root}.
A vertex common to two edge-disjoint subtrees
of a rooted tree is a root in at least one of the subtrees.
Therefore every cut vertex in $G$ is a BCC root.
Conversely, a BCC root $u$ is also a cut vertex in $G$
unless $u$ is a root in $F$ with only one child.
Every BCC of $G$ consists
precisely of the vertices in a
particular BCC subtree $F'$ and the edges in $G$
that join two such vertices, i.e.,
whose lower endpoint (with respect to~$F$) lies in~$F'$
but is not the root of~$F'$.
An edge is a bridge exactly if, together with
its endpoints, it constitutes a full BCC subtree.
A 2ECC, finally, is either such a 1-edge BCC subtree
or a maximal connected subgraph of $G$
with at least one edge and without bridges.

For each $w\in V$, denote by $P(w)$ the assertion that
$w$ has a parent $v$ in $F$ and $G$ contains at
least one edge between a descendant of $w$ and a proper
ancestor of~$v$.
If $\{v,w\}$ is an edge in $F$ and $v$ is the
parent of~$w$,
$P(w)=\Tvn{false}$ exactly if $v$ is the root
of the BCC subtree that contains $\{v,w\}$.
We can compute $P(w)$ for all $w\in V$ by
initializing all entries in a Boolean array
$Q[1\Ttwodots n]$ to \Tvn{false} and
processing all nontree edges as follows:
To process a nontree edge $\{x,y\}$,
where $y$ is a descendant of $x$,
start at $y$ and follow
the path in $F$ from $y$ to $x$, setting
$Q[z]:=\Tvn{true}$ for every vertex $z$ visited,
but omitting this action for the last two
vertices (namely $x$ and a child of~$x$).
Suppose that we process each nontree edge $\{x,y\}$,
where $y$ is a descendant of~$x$, when a preorder
traversal of $F$ reaches $x$ and before it
proceeds to children of~$x$.
Then we can stop the processing
of $\{x,y\}$ once we reach a vertex $z$ for
which $Q[z]$ already has the value \Tvn{true}---the
same will be true for all outstanding vertices~$z$.
Therefore the processing of all nontree edges
can be carried out in $O(n+m)$ time,
after which $Q[w]=P(w)$ for all $w\in V$.
To solve one of the problems considered in this
subsection, compute the DFS forest $F$
and traverse it to compute $Q$, as just described,
while executing the following additional
problem-specific steps:

\emph{Cut vertices:}
Output each vertex $v$ in $F$ that
has a child $w$ with $P(w)=\Tvn{false}$
and is not a root in $F$ or has two or more children.

\emph{Bridges:}
Output each tree edge $\{u,v\}$, where $u$ is the parent
of $v$, for which $P(v)=\Tvn{false}$ and
$P(w)=\Tvn{false}$ for every child $w$ of~$v$.

\emph{Biconnected components:}
Specialize the traversal of $F$ to always visit a
vertex $w$ with $P(w)=\Tvn{false}$ before a sibling
$w'$ of $w$ with $P(w')=\Tvn{true}$.
To compute the biconnected components of $G$ with
their vertices and edges, when the traversal
withdraws over a tree edge $\{v,w\}$ from a vertex $w$
to its parent~$v$, output the edges in~$G$
that have $w$ as their lower endpoint
(including $\{v,w\}$), output $w$
itself and, if $P(w)=\Tvn{false}$,
also output $v$ and \emph{wrap up}
the current BCC, i.e., except in the case of the
very last component, output a component separator
or increment the component counter.
Visiting the children $w$ of a vertex $v$ with
$P(w)=\Tvn{true}$ after those with $P(w)=\Tvn{false}$
ensures that the vertices and edges of the BCC
that contains $\{v,w\}$
are output together for each $w$ without intervening
vertices and edges of other biconnected components.

\emph{2-edge-connected components:}
Specialize the traversal of $F$ so that for
each vertex $v$, a child $w$ of~$v$
for which $\{v,w\}$ is a bridge is always
visited before a
child $w'$ of $v$ for which $\{v,w'\}$ is
not a bridge.
Suppose that the traversal withdraws from a
vertex $w$ to its parent~$v$.
If $w$ has at least one incident edge that
is not a bridge, output~$w$ and, if 
$\{v,w\}$ is a bridge, wrap up the current 2ECC.
If $\{v,w\}$ is a bridge, 
output $v$, $w$ and $\{v,w\}$ and
wrap up the current 2ECC.
If $\{v,w\}$ is not a bridge, 
output all nonbridge edges of which $w$ is
the lower endpoint (including $\{v,w\}$).
Finally, when the traversal withdraws from
a root $u$ with at least one incident
edge that is not a bridge,
output $u$ and wrap up the current 2ECC.
As above,
visiting those children of a given vertex $v$
that are adjacent to $v$ via bridges before
the other children of~$v$
ensures that the vertices and edges
of the 2ECC that contains several
edges incident on $v$, if any,
are output together without intervening vertices
and edges of other 2-edge-connected components.

When the traversal of $F$ reaches a vertex $v$
with a child $w$, $Q[w]$ will have reached its final
value, $P(w)$, and will never again be written to.
It is now obvious that we can test
at that time whether $v$ is a cut vertex and
whether the edge between $v$ and its parent in $F$,
if any, is a bridge in $O(d+1)$ time,
where $d$ is the degree of~$v$.
It follows that each of the four problems
considered above can be solved in $O(n+m)$ time.
In the most complicated case, that of
2-edge-connected components, in order
to test during the processing of a vertex~$v$
whether an edge $\{v,w\}$ is a bridge,
where $w$ is a child of~$v$ in~$F$,
carry out a ``preliminary visit''
of the children of~$w$ in~$F$.

We can compute the DFS forest $F$ with the algorithm of
Lemma~\ref{lem:parent}, which needs
$(n+2 L_1(G))(1+O({{\log\log n}/{\log n}}))$
bits plus $n$ bits that can be reused.
The subsequent traversal needs $n$ bits for the array~$Q$.
In addition, when the computation of $Q$
described above
processes a nontree edge $\{x,y\}$, it needs to
know whether $y$ is an ancestor
or a descendant of~$x$.
We can use another Boolean array $A[1\Ttwodots n]$ to handle this
issue, ensuring for each $y\in V$ that at all times
$A[y]=\Tvn{true}$ if and only if $y$ is an ancestor
of the current vertex of the traversal of $F$
(i.e., if $y$ is gray).
In some cases, however, we can make do with less space.
Observe that in the computation of cut vertices and
bridges, the value of $Q[v]$ is never again used
after the arrival of the traversal of $F$ at~$v$.
When the traversal reaches $v$ and $Q[v]$ has been
inspected, we can therefore set $Q[v]:=\Tvn{true}$
without detriment to the use of $Q$.
Suppose that when processing a nontree edge $\{x,y\}$
in the computation of $Q$, we consult $Q[y]$
instead of $A[y]$ to know whether $y$ is an ancestor
of the current vertex~$x$.
If $A[y]=\Tvn{true}$, the artificial change to $Q$
introduced above ensures that we necessarily also
have $Q[y]=\Tvn{true}$, so that the algorithm
proceeds correctly.
If $A[y]=\Tvn{false}$, we may have $Q[y]=\Tvn{true}$,
in which case the processing of $\{x,y\}$ stops
immediately, but then $P(y)=\Tvn{true}$
($y$ has not yet been reached by the traversal,
and so $Q[y]$ was not set artificially to \Tvn{true})
and it is correct to do nothing.

If our goal is to compute the biconnected or
2-edge-connected components of $G$ with their
vertices, but not with their edges, we are
in an intermediate situation:
We need to distinguish between three different
combinations of $A[y]$ and $Q[y]$
(now with the original $Q$), but if $Q[y]=\Tvn{true}$
the value of $A[y]$ is immaterial as above,
and $Q[y]$ never changes from
\Tvn{true} to \Tvn{false}.
We can therefore represent $A$ and $Q$ together
through a ternary array with $n$ entries.
Altogether, we have proved the following result.

\begin{theorem}
\label{thm:fourdep}
Given an undirected
graph $G$ with $n$ vertices and $m$ edges,
we can compute the following in $O(n+m)$ time
and with the number of bits indicated:
\begin{itemizewithlabel}
\item[(a)]
The cut vertices and bridges of $G$ with
$(2 n+2 L_1(G))(1+O({{\log\log n}/{\log n}}))$
$=2 n\log(2+{{8 m}/n})(1+O({{\log\log n}/{\log n}}))$ bits;
\item[(b)]
The biconnected and 2-edge-connected components of $G$
with their vertices with
$((1+\log_2 3)n+2 L_1(G))(1+O({{\log\log n}/{\log n}}))$ bits;
\item[(c)]
The biconnected and 2-edge-connected components of $G$
with their edges and possibly vertices with
$(3 n+2 L_1(G))(1+O({{\log\log n}/{\log n}}))$ bits.
\end{itemizewithlabel}
\end{theorem}

Kammer et al.~\cite{KamKL16} consider the problem of
preprocessing an undirected graph $G$ so as later to
be able to output the vertices and/or edges of a
single BCC, identified via one of its edges, in
time at most proportional to the number
of items output.
Having available the parent pointers of a DFS forest $F$
and the array $Q$ corresponding to~$F$, we can solve the
problem in the following way, which is the
translation of the procedure of
Kammer et al.\ to our setting:
Given a request to output the BCC $H$ that contains
an edge $\{x,y\}$, first follow parent pointers
in parallel from $x$ and $y$ until one of the searches
hits the other endpoint or a root in~$F$.
This allows us to determine which of $x$ and $y$
is an ancestor of the other vertex in time at
most proportional to the number of items to be output.
Then traverse the subtree of $F$ reachable from the
lower endpoint of $\{x,y\}$ without crossing any
edge between a BCC root and its single child,
producing the same output at each
vertex as described above for the output of all
biconnected components.
In addition, at the uniquely defined vertex
$w$ with $P(w)=\Tvn{false}$ visited by the
search, also output the parent $v$ of~$w$
(but without continuing the traversal from~$v$).
In order to carry out this procedure efficiently, we need a
way to iterate over the edges in $F$ incident
on a given vertex $w$ and, if the edges of $H$
are to be output, over the nontree edges that
have $w$ as their lower endpoint.
To this end
we can equip each vertex $w$ of degree $d$ with
a choice dictionary
(Lemma~\ref{lem:choice}) for a universe size of~$d$
that allows us to iterate over the relevant
edges in time at most proportional to their number.
This needs another $2 m+O({m/{\log n}})$ bits.
Very similar constructions allow us to output
the vertices and/or the edges of a single
2-edge-connected component.

\begin{theorem}
There is a data structure that can be initialized
for an undirected graph $G$ with $n$ vertices
and $m$ edges in $O(n+m)$ time, subsequently allows
the vertices and/or the edges of the biconnected
or 2-edge-connected component that
contains a given edge to be output in time
at most proportional to the number of items
output, and uses
$(3 n+2 m+2 L_1(G))(1+O({{\log\log n}/{\log n}}))$
bits.
\end{theorem}

\section{The Density-Independent Case}

\subsection{Depth-First Search}

Some aspects of the following proof are
similar to those of
\cite[Lemma~3.2]{ElmHK15}.

\begin{theorem}
\label{thm:loglog}
A DFS of a directed or undirected graph with $n$
vertices and $m$ edges can be carried out in $O(n+m)$
time with $O(n\log\log(4+{m/n}))$ bits of
working memory.
\end{theorem}

\begin{proof}
Assume without loss of generality that
$m\ge{n/2}\ge 1$.
We simulate the algorithm of
Theorem~\ref{thm:dfs}, but using
asymptotically less space
(unless $m=O(n)$).
Recall that the algorithm employs a stack $S$
whose size is always bounded by $n r$, where
$r=O(\log(2+{m/n}))$.
When a vertex is discovered by the DFS and
enters $S$, it is permanently assigned
an integer \emph{hue}.
The first vertices to be discovered are given
hue~1, the next ones receive hue~2, etc.,
and the vertices on $S$ with a common hue
are said to form a \emph{segment}.
In general, a new segment is begun whenever the
current segment for the first time occupies
more than $n$ bits on~$S$.
Thus no hue larger than $r$ is ever assigned.

As in \cite{ElmHK15}, the algorithm does not
actually store $S$, which is too large, but only
a part $S'$ of $S$ consisting of the one or two
segments at the top of~$S$.
When a new segment is begun and $S'$ already
contains two segments, the older of these is
first dropped to make room for the new segment.
By construction, $S'$ always occupies $O(n)$ bits.

The algorithm operates as that of Theorem~\ref{thm:dfs},
using $S'$ in place of $S$, except when a pop
causes $S$ but not $S'$ to become empty.
Whenever this happens the top segment of $S$ is
\emph{restored} on $S'$, as explained below,
after which the DFS can resume.
Between two stack restorations a full segment
disappears forever from $S$, so the total number of stack
restorations is bounded by~$r$.

In order to enable efficient stack restoration,
we maintain for each vertex $v$
(a) its color---white, gray or black;
(b) its hue;
(c) whether it is currently on $S'$;
(d) the number of groups of $\Tceil{m/n}$
(out)edges incident on $v$
that have been explored with $v$ as the current vertex.
The number of bits needed
is $O(1)$ for items (a) and (c),
$O(\log(2+r))$ for item (b) and $O(\log(2+{{d n}/m}))$ 
for item (d), where $d$ is the degree of $v$.
Summed over all vertices, this yields a bound
of $O(n\log(2+r))$ bits, as required.
For each segment $J$ on $S$, we also store
on a second stack $S_{\mathrm{t}}$ the last
vertex $u$ of $J$ (the vertex closest to the top of~$S$)
and the number of (out)edges incident on $u$
explored by the DFS with $u$ as the current vertex.
The space occupied by $S_{\mathrm{t}}$ is negligible.

To restore a segment $J$, we push the
bottommost
entry of $J$ on $S'$ and initialize accordingly
the variables kept
outside of $S'$ to interpret entries of $S'$ correctly.
This can be done in constant time by consulting
either the entry on $S_{\mathrm{t}}$ immediately below the top
entry or separately remembered information
concerning the root of the current DFS tree.
We proceed to push on $S'$ the remaining
vertices in $J$ one by one,
stopping when the top entries on $S'$ and
$S_{\mathrm{t}}$ agree,
at which point the restoration of $J$ is complete
and the normal DFS can resume.
Each entry on $S'$ above that of
a vertex $u$ is found by determining the
first gray vertex in $u$'s adjacency array
(counted cyclically from the position of $u$'s parent)
that belongs to~$J$ (as we can tell from its hue)
and is not already on~$S'$.

Because of item (d) of the information kept for
each vertex, the search in the adjacency array
of a vertex of degree $d$ is easily made to spend
$O(\min\{d+1,\Tceil{{m/n}}\})$ time
on entries that were inspected before.
Over all at most $r$
restorations and over all vertices
of degree at most $\sqrt{{m/n}}$, this sums
to $O(r n\sqrt{{m/n}})
=O(m)$.
Since a restoration involves
$O({n/r})$ vertices of degree
larger than $\sqrt{{m/n}}$, the sum over
all restorations and over all such vertices
is $O(r({n/r})({m/n}))=O(m)$.
Altogether, therefore, the algorithm spends
$O(m)$ time on restorations and $O(m)$ time
outside of restorations.
\end{proof}

Our remaining algorithms depend on the following
lemma, which can be seen as a weak dynamic version
of Lemma~\ref{lem:static}.

\begin{lemma}
\label{lem:ragged}
For all $n,N\in\TbbbN$, following an
$O(n)$-time initialization,
an array of $n$ initially empty binary strings
$s_1,\ldots,s_n$ that at all times satisfy
$|s_i|=O(\log n)$ for $i=1,\ldots,n$
and $\sum_{i=1}^n |s_i|\le N$ can be maintained
in $O(n\log\log n+N)$ bits under constant-time reading
and amortized constant-time writing of
individual array entries.
\end{lemma}

\begin{proof}
Compute a positive integer $h$
with $h=\Theta((\log n)^2)$ and partition the strings
into $O({n/h})$ \emph{groups} of
$h$ strings each, except that the last
group may be smaller.
For each group, store the strings
in the group in $O(\log n)$ \emph{piles}, 
each of which holds all strings of one
particular length in no particular order.
For each string, we store its length
and its position within the corresponding pile.
Conversely, the entry for a string on a pile, besides
the string itself, stores the number of the
string within its group.
The size of the bookkeeping information amounts
to $O(\log\log n)$ bits per string and
$O(n\log\log n)$ bits altogether.

When a string changes, the string may need
to move from one pile to another within its group.
This usually leaves a ``hole'' in one pile, which is
immediately filled by the entry that used to be
on top of the pile.
This can be done in constant time, which also covers
the necessary update of bookkeeping information.

We are now left with the problem of representing
$O(({n/h})\log n)=O({n/{\log n}})$ piles.
Divide memory into \emph{words} of $\Theta(\log n)$
bits, each of which is large enough to hold one of the
strings $s_1,\ldots,s_n$.
Rounding upwards, assume that each pile at all
times occupies an integer number of words---this
wastes $O(n)$ bits.
The update of a string may cause the sizes of up to
two piles to increase or decrease by one word,
but
no operation changes
the size of a pile by more than one word.
Each pile is stored in a \emph{container},
of which it occupies at least a quarter.
For each pile we maintain its size, the size of
its container, and the location in memory of its
container, a total of $O(\log n)$ bits per pile
and $O(n)$ bits altogether.
We also maintain in a \emph{free pointer} the
address of the first memory word after the last container.

When a pile outgrows its container, a new container,
twice as large, is first allocated for it starting
at the address in the free pointer,
which is incremented correspondingly.
The pile is moved to its new container, after
which its old container is considered \emph{dead}.
Conversely, when a pile would occupy less than a quarter
of its container after losing a string, the pile
is first moved to a new container
of size twice the size of the pile
after the operation and
also allocated from the address
in the free pointer.
If every operation that operates on a pile
pays 5 coins to the pile, by the time
when the pile needs to migrate to a new
container, it will have accumulated enough
coins to place a coin on every position in
the old container and a coin on every element
of the pile.
In terms of an amortized time bound,
the latter coins can pay for the migration
of the pile to its new container.

When an operation would cause
the size of the dead containers to exceed that
of the live containers
(the containers that are currently in use) plus~$n$ bits, we
carry out a ``garbage collection'' that eliminates
the dead containers and re-allocates the piles in
tightly packed new live containers
in the beginning of
the available memory, where each new container
is made twice as large as the pile that it contains,
and the free pointer is reset accordingly.
The garbage collection can be paid for by the
coins left on dead containers.

A string can be read in constant time, and updating
it with a new value takes constant amortized time,
as argued above.
Because every pile occupies at least a quarter of
its (live) container and the dead containers are never
allowed to occupy more space than the live containers,
plus $n$ bits,
the number of bits occupied by the array of strings
at all times is $O(n+N)$.
\end{proof}

\begin{corollary}
\label{cor:ragged}
For all $n,N\in\TbbbN$, following an
$O(n)$-time initialization,
an array of $n$ initially empty binary strings
$s_1,\ldots,s_n$ that at all times satisfy
$|s_i|=O({{\log n}/{\log\log n}})$ for $i=1,\ldots,n$
and $\sum_{i=1}^n |s_i|\le N$ can be maintained
in $O(n+N)$ bits under constant-time reading
and amortized constant-time writing of
individual array entries.
\end{corollary}

\begin{proof}
Maintain groups of
$\Theta(\log\log n)$ strings with the
data structure of the previous lemma.
In more detail,
compute a positive integer $q$ with
$q=\Theta(\log\log n)$ and partition the $n$ strings
$s_1,\ldots,s_n$ into $\Tceil{{n/q}}$ \emph{blobs}
of $q$ consecutive strings each, except that the
last blob may be smaller.
If a blobs consists of the strings $s_i,\ldots,s_j$,
let its \emph{label} be the binary string
$\texttt{1}^{|s_i|}\texttt{0} s_i\cdots
\texttt{1}^{|s_j|}\texttt{0} s_j$.
Because the label of a blob is of $O(\log n)$ bits,
given the label and the number of a string $s$
within the blob, we can extract $s$ in constant
time by lookup in a table of $O(n)$ bits that can
be computed in $O(n)$ time.
Similarly, given a new value for $s$, we can update
$s$ within the blob in constant time.
Maintain the sequence of $\Tceil{{n/q}}$ blobs of
$O(\log n)$ bits each with the data structure of
Lemma~\ref{lem:ragged}.
The number of bits needed is
$O(({n/q})\log\log n+N)=O(n+N)$,
and the operation times are as claimed.
\end{proof}

\begin{theorem}
\label{thm:logstar}
A DFS of a directed or undirected graph with $n$
vertices and $m$ edges can be carried out in
$O(n+m\Tlogstar n)$ time with $O(n)$ bits of
working memory.
\end{theorem}

\begin{proof}
Assume without loss of generality that
$m\ge{n/2}\ge 1$.
Compute a positive integer $t$ and a sequence
$p_1,\ldots,p_t$ of $t$ powers of~2 with
the following properties:

\begin{itemize}
\item[$\bullet$]
$1=p_t<p_{t-1}<\cdots<p_2<p_1=\Theta(\sqrt{\log n})$.
\item[$\bullet$]
For $i=2,\ldots,t$, $p_i\ge\log p_{i-1}$.
\item[$\bullet$]
$t=O(\Tlogstar n)$.
\end{itemize}

This is easy:
Begin by computing $p_1$ as a power of~2 with
$p_1=\Theta(\sqrt{\log n})$, at least~1, take
$t=1$ and then, as long as $p_t>1$,
increment $t$ and let $p_t$ be the smallest
power of~2 no smaller than $\log p_{t-1}$,
i.e., $p_t=2^{\Tceil{\log\log p_{t-1}}}$.

As in the algorithm of
Theorem~\ref{thm:loglog}, we operate with a
conceptual stack $S$ and an actual stack $S'$
that contains the topmost one or two segments on~$S$.
By definition, a complete segment now contains exactly
$\Tceil{{n/{p_1^2}}}$ vertices.
The complete segments in turn are partitioned into
\emph{stripes}, each of which has a \emph{rank}
drawn from $\{1,\ldots,t\}$.
For $i=1,\ldots,t$, a stripe of rank $i$
comprises exactly $({{p_1}/{p_i}})^2$ segments
and therefore at least ${n/{p_i^2}}$ vertices.
A stripe of rank~1 can be identified with
the single segment that it contains.
From the bottom to the top of $S$, the stripes
occur in an order of nonincreasing rank
(informally, larger stripes are deeper
in the stack).
Taken in the same order, the stripes are also assigned the
indices $0,1,2,\ldots,$ and each vertex is marked
with the (current) index of its stripe.
For $i=2,\ldots,t$, since the number of stripes
of rank at least $i-1$ is bounded by
$p_{i-1}^2$, the index of every such stripe is
an integer of at most
$2\log p_{i-1}\le 2 p_i$ bits.

An additional stack
records
the ranks of all stripes in the order in
which the stripes occur on~$S$.
Using this stack and $S\Tsub t$ to start
at the bottom entry
of the topmost stripe, we can step
through the vertices of that stripe in the
order in which they occur on~$S$.
If the stripe is of rank~1, in particular, this
enables us to carry out a stack restoration.
If the topmost stripe is of rank $i>1$, we can carry
out the stack restoration by \emph{splitting}
the stripe into $({{p_{i-1}}/{p_i}})^2$
stripes of rank $i-1$ and proceeding recursively.
Conversely, if at some point there are
$2({{p_{i-1}}/{p_i}})^2$ stripes of rank $i-1$
for $i=2$, we \emph{join} the bottommost
$({{p_{i-1}}/{p_i}})^2$ of these
into a stripe of rank~$i$ and continue the
conditional join recursively for $i=3,\ldots,t$.
In order to know when to stop the join,
maintain for $i=1,\ldots,t$ the number of stripes
of rank~$i$.

For $i=2,\ldots,t$, there are never more than
$2({{p_{i-1}}/{p_i}})^2$ stripes of rank $i-1$,
and the number of vertices contained in
stripes of rank $i-1$ is
$O(({n/{p_{i-1}^2}})({{p_{i-1}}/{p_i}})^2)
=O({n/{p_i^2}})$.
As noted above, every index of such a stripe
is of $O(p_i)$ bits, 
so the total number of bits consumed by indices
of stripes of rank $i-1$ is $O({n/{p_i}})$.
Summed over all values of $i$, this yields
$O(n)$ bits occupied by stripe indices.
Since the stripe indices are of
$O(\log\log n)$ bits by the upper
bound on $p_1$, they can be maintained
with the data structure of
Corollary~\ref{cor:ragged}
at a cost of constant time per operation
and a total of $O(n)$ bits.
The other data structures used by
the algorithm are easily seen to
fit in $O(n)$ bits as well.

Fix $i\in\{2,\ldots,t\}$, let $N_0$ be the number
of vertices in a single stripe of rank~$i$,
let $N$ be the current number of vertices
in stripes of rank at most $i-1$
and define $\Phi$ as $|N-N_0|$.
The split of a stripe of rank $i$ into stripes
of rank $i-1$ reduces $\Phi$ from $N_0$ to~0,
the join of stripes of rank $i-1$ into
a stripe of rank~$i$ reduces $\Phi$
by exactly~$N_0$, and a push or pop on $S$
increases $\Phi$ by at most~1.
A potential argument now shows the sum over all
splits of stripes of rank $i$ or joins of stripes
to stripes of rank $i$ of the number of vertices
contained in the stripes concerned
and the stripes above them to be $O(n)$.
Summing over all values of~$i$, this yields a bound
of $O(t n)$.
If we again store with each vertex the number of
completely explored groups of $\Tceil{{m/n}}$ incident
(out)edges, the time needed for all splits and
joins is $O(t n({m/n}))=(m\Tlogstar n)$,
and all other parts of the algorithm
work in $O(m)$ time.
\end{proof}

By using only the ranks $1,\ldots,k$, for some
$k\in\TbbbN$, i.e., by omitting all joins
that would create stripes of rank $k+1$ or more,
we can lower the time bound
of the previous algorithm to $O(n+k m)$,
but at the price of having to store
indices of stripes of rank $k$ for
almost all vertices.

\begin{theorem}
\label{thm:dfs-k}
For every constant $k\in\TbbbN$,
a DFS of a directed or undirected graph with $n$
vertices and $m$ edges can be carried out in
$O(n+m)$ time with $O(n\log^{(k)}\! n)$ bits of
working memory.
\end{theorem}

\subsection{Biconnected and 2-Edge-Connected Components}

This subsection describes algorithms for the problems
considered in Theorem~\ref{thm:fourdep}, but with
the resource bounds of Theorems \ref{thm:logstar}
and~\ref{thm:dfs-k}.
The approach uses elements of an algorithm of
Kammer et al.~\cite{KamKL16} for computing
cut vertices.

The main differences to the density-dependent setting
of Theorem~\ref{thm:fourdep}
can be explained in terms on a nontree edge
$\{x,y\}$, where $y$ is a descendant of~$x$.
In the density-dependent case,
because a DFS forest is computed in a first pass
and remembered, in a second pass $\{x,y\}$ can
be processed when a preorder traversal
of the forest reaches~$x$.
This makes the processing of $\{x,y\}$,
which is essentially the marking of a number of
vertices, efficient, because the marking can stop
as soon as it reaches a vertex that is already marked.

In the density-independent setting there is only
a single pass, and $\{x,y\}$ must be processed when
it is explored by the DFS in that pass, at which
time the current vertex of the DFS is~$y$.
The vertices to be marked all lie on the gray path
between $x$ and $y$, but they may be few and
distributed irregularly, and we cannot afford
to follow the gray path backwards all the way from $y$
to $x$ in order to find them.
To solve this problem we could keep
on a separate stack $S\Tsub u$ only those gray vertices
that have not yet been marked and let the
processing of $\{x,y\}$ pop and mark vertices
from $S\Tsub u$ until, loosely speaking,
$x$ is reached.

The approach outlined in the previous paragraph
would be correct, but the stack $S\Tsub u$ of
unmarked gray vertices would take up too much space.
Indeed, just as we can keep only one or two segments,
the \emph{surface segments},
of the virtual stack $S$ on an actual stack $S'$.
we can keep only the part of $S\Tsub u$ that
contains vertices in the surface segments on
an actual stack that, for convenience, we
continue to call $S\Tsub u$.
Because of this, we cannot mark vertices in
the other segments, the \emph{buried segments}.
The marking must still be carried out, but it can be postponed
until the segments in question once again become
surface segments.
In order to realize this, we use a special
``propagating mark'' that marks a vertex but also
calls for the marking to be extended towards
the top of the stack once the propagating mark
is no longer buried
(other, normal, marks may have become buried and
should not be propagated in the same way).

The ``scope'' of a propagating mark could extend
to the vertex $y$ that was the current vertex
of the DFS when the propagating mark was placed.
Since this is not easy to handle, we instead
stipulate that the ``scope'' of a propagating
mark ends at the end of its stripe.
As a consequence, the processing of a nontree
edge $\{x,y\}$ may require several buried
stripes to receive propagating marks.
For efficiency reasons we must prevent
stripes for which this already happened
from being processed again, which we can do
by placing buried stripes, represented
by their bottommost vertices, on $S\Tsub u$,
processing stripes only when they
are on $S\Tsub u$, and removing them from
there when this happens.
Since the number of stripes is small,
storing distinct stripes on $S\Tsub u$ does not
violate the space bound.

The remaining details are given in the following proof.
In particular, it is shown how one can know when
to stop when popping from~$S\Tsub u$.

\begin{theorem}
Given an undirected
graph $G=(V,E)$ with $n$ vertices and $m$ edges,
in $O(n+m\Tlogstar n)$ time and with
$O(n)$ bits of working memory, we can
compute the biconnected components and the
2-edge-connected components with their
vertices and/or their edges,
the cut vertices and the bridges of~$G$.
For every constant $k\in\TbbbN$, the same problems
can also be solved in $O(n+m)$ time
with $O(n\log^{(k)}\! n)$ bits of
working memory.
\end{theorem}

\begin{proof}
Let us change the DFS algorithm of
Theorem~\ref{thm:logstar}
or Theorem~\ref{thm:dfs-k} to make it carry out
a stack restoration not when the actual
stack $S'$ is empty, but already when it contains
only two vertices
(without loss of generality, full segments
comprise at least three vertices).
Let us call this modified stack restoration
\emph{eager restoration}.
We show how to augment the algorithm with
steps that maintain a
Boolean array $R[1\Ttwodots n]$
such that whenever the gray path of the DFS
contains three vertices $u$, $v$ and $w$
in that order, $R[u]=\Tvn{true}$ exactly
if the part of the graph explored so far
contains a cycle through the two
edges $\{u,v\}$ and $\{v,w\}$.
With $P(w)$ defined as in
Subsection~\ref{subsec:dep-apl}, it is
clear that $P(w)$ can be read off $R[u]$
when the search is about to withdraw
over $\{v,w\}$.
This enables us to compute a table of $P$,
and the arguments given in
Subsection~\ref{subsec:dep-apl} show
how to solve the problems indicated in
the theorem using only $O(n+m)$ additional
time and $O(n)$ additional bits.

$R$ is in fact a virtual array implemented
via an actual array $\overline{R}[1\Ttwodots n]$,
each of whose entries takes values in
$\{\Tvn{false},\Tvn{true},\Tvn{propagating-true}\}$.
The connection is as follows:
For $u\in V$, $R[u]=\Tvn{true}$ if and only if
$\overline{R}[u]\in\{\Tvn{true},\Tvn{propagating-true}\}$
or $\overline{R}[u']=\Tvn{propagating-true}$ for some
$u'\in V$ that belongs to the same stripe as $u$
and precedes $u$ within that stripe
(thus $u'$ was pushed before~$u$).
Moreover, we stipulate that
$R[u]=\overline{R}[u]$ for all $u$
currently stored on $S'$ (i.e., belonging
to a surface segment), which implies that we can
determine $R[u]$ in constant time whenever we need it.

For each $u\in V$ currently on $S'$, we store with
$u$ its position on $S'$.
We also use an additional stack $S_{\mathrm{u}}$ that contains
some of the vertices on $S$ in the same order
as on $S$.
If $u$ is stored on $S'$, it is also present on $S_{\mathrm{u}}$
exactly if it is followed on the gray path by
at least two vertices and
$R[u]=\Tvn{false}$---informally, if $u$'s gray
child has not yet been ``covered'' by a nontree edge.
A vertex $u$ on $S$ but not on $S'$ can be stored on
$S_{\mathrm{u}}$
only if it is the bottommost vertex in its stripe;
if so, it is certain to be stored on $S_{\mathrm{u}}$ if
$R[u']=\Tvn{false}$ for at least one vertex $u'$
in the same stripe as~$u$.
Informally, $u$ now represents its stripe and
records the fact that the stripe may contain one
or more ``uncovered'' vertices.

What remains is to describe the manipulation of
$\overline{R}$ and $S_{\mathrm{u}}$,
which must respect the invariants
introduced above.
Initially $\overline{R}[u]=\Tvn{false}$ for all
$u\in V$ and $S_{\mathrm{u}}$ is empty.
When the DFS discovers a vertex $w$ over a tree
edge $\{v,w\}$ and $v$ has a parent $u$,
we set $\overline{R}[u]:=\Tvn{false}$
and push $u$ on $S_{\mathrm{u}}$.
When the DFS withdraws from a vertex $w$ to its
parent $v$ and $v$ has a parent $u$,
we
pop $u$ from the stack
$S_{\mathrm{u}}$ if it is present there---if so,
it is the top entry.
When processing a nontree edge $\{x,y\}$, where
$y$ is the current vertex and
$x$ is an ancestor of $y$,
we pop all vertices
from $S_{\mathrm{u}}$ that are equal to $x$ or closer
than $x$ to the top of $S$.
Informally, some of these vertices---those on $S'$---represent
only themselves, while each of the remaining vertices
represents a whole stripe.
Because we know the position on $S'$ of every
vertex stored on $S'$ and the stripe index
of every vertex,
the process can happen in constant time plus
constant time per vertex popped.
For each vertex $u$ popped, we set
$\overline{R}[u]:=\Tvn{true}$ if $u$ is stored
on $S'$ (if $u$ belongs to a surface segment)
and $\overline{R}[u]:=\Tvn{propagating-true}$ if not.
When a stripe is split or several stripes are joined,
the entries on $S_{\mathrm{u}}$ are changed correspondingly:
The bottommost vertex of each stripe that disappears
is popped from $S_{\mathrm{u}}$
(if present), and the bottommost vertex of each
new stripe is pushed on $S_{\mathrm{u}}$.

When a segment is restored, an invariant demands
that the value
\Tvn{propagating-}
\Tvn{true} be eliminated from
its vertices.
This is done in a simple scan of the segment
from bottom to top:
If the value $\Tvn{propagating-true}$ is ever
encountered, the value of $\overline{R}$
is set to $\Tvn{true}$ for the relevant vertex
and all vertices that follow it, in accordance
with the relation between $\overline{R}$ and
$R$ described above.
Moreover, all vertices $u$ with $R[u]=\Tvn{false}$
after the bottommost vertex
are pushed on $S_{\mathrm{u}}$
(because of the eager restoration,
every such vertex is followed on the gray path
by at least two vertices).
When a segment is dropped from $S'$, all entries
of its vertices on $S_{\mathrm{u}}$ are replaced by
a single entry for its bottommost vertex.

The mapping of the $O({n/{\log n}})$
vertices on $S'$ to their positions on $S'$ can be
maintained in a way described by Kammer
et al.~\cite{KamKL16}:
When the first vertex of a segment is pushed
on $S'$, the DFS is first executed without
this positional information until it has computed
the set $U$ of vertices in the segment.
Then $\Theta(\log n)$ bits are allocated
to each vertex in $U$ using static space
allocation, and finally the relevant part of
the DFS is repeated, at which point the
position of each new vertex on $S'$ can
be recorded.
A similar procedure is followed when
a segment is restored.
It is shown in~\cite{KamKL16} how table
lookup allows the static
space allocation to happen sufficiently
fast, namely in $O({n/{\log n}})$ time---in
essence, it suffices to mark each of
$O({n/{\log n}})$ regularly spaced vertices in $V$ with
the number of smaller vertices belonging to~$U$.
An alternative is to appeal to the fast construction
of rank-select structures of
Baumann and Hagerup~\cite{BauH17}.
In either case, it is easy to see that the total number
of bits needed is $O(n)$.
The fact that the usual stack restoration
has been replaced by eager restoration
does not invalidate the bound established in its proof
on the time needed for splitting and joining stripes.
The other steps described above are no more
expensive, to within a constant factor.
In particular, the number of pops
from $S_{\mathrm{u}}$ is bounded by the number of
pushes on $S_{\mathrm{u}}$.
Therefore the asymptotic time and space bounds
demonstrated for the DFS are valid also for
the entire computation.
\end{proof}

\bibliography{all}

\end{document}